\documentclass [winedt,yap]{iitparc}
\usepackage{cite}
\usepackage{amsmath,amssymb,amsfonts,bm}
\usepackage{eufrak}
\usepackage{lscape}
\usepackage{floatfig,wrapfig,epsfig}
\usepackage{subfigure}
\usepackage{color}
\usepackage{psboxit}
\usepackage{rotating}
\usepackage{curves}
\usepackage{ulem}
\usepackage[mathscr]{eucal}
\RequirePackage{psfrag}
\RequirePackage{graphicx}

\begin{document}\normalem
\initfloatingfigs
\frontmatter          

\IssuePrice{25.00}%
\TransYearOfIssue{2011}%
\TransCopyrightYear{2011}%
\OrigYearOfIssue{2011}%
\OrigCopyrightYear{2011}%

\TransVolumeNo{72}%
\TransIssueNo{6}%
\OrigIssueNo{20}%


\mainmatter

\setcounter{page}{345}
\CRubrika{CONTROL SCIENCES} \Rubrika{CONTROL SCIENCES}
\OrigJournalName{Problemy Upravleniya} \OrigIssueNo{1}
\OrigYearOfIssue{2010}
\OrigCopyrightYear{2010}%
\TransYearOfIssue{2011}%
\TransCopyrightYear{2011}%
\TransIssueNo{6} 

\newlength{\widebarargwidth}
\newlength{\widebarwidth}
\newlength{\widebarargheight}
\newlength{\widebarargdepth}
\DeclareRobustCommand{\widebar}[1]{%
  \settowidth{\widebarargwidth}{\ensuremath{#1}}%
  \settoheight{\widebarargheight}{\ensuremath{#1}}%
  \settodepth{\widebarargdepth}{\ensuremath{#1}}%
  \addtolength{\widebarargwidth}{-0.2\widebarargheight}%
  \addtolength{\widebarargwidth}{-0.2\widebarargdepth}%
  \makebox[0pt][l]{\addtolength{\widebarargheight}{0.3ex}%
    \hspace{0.2\widebarargheight}%
    \hspace{0.2\widebarargdepth}%
    \hspace{0.5\widebarargwidth}%
    \setlength{\widebarwidth}{0.6\widebarargwidth}%
    \addtolength{\widebarwidth}{0.3ex}%
    \makebox[0pt][c]{\rule[\widebarargheight]{\widebarwidth}{0.1ex}}}%
  {#1}}

\title{Voting in a Stochastic Environment: The Case of Two Groups}

\author{P. Yu. Chebotarev,  A. K. Loginov, Ya. Yu. Tsodikova,\\
Z. M. Lezina,  and V. I. Borzenko}

\institute{Trapeznikov Institute of Control Sciences, Russian Academy of Sciences, Moscow, Russia}

\titlerunning{VOTING IN A STOCHASTIC ENVIRONMENT}

\authorrunning{Chebotarev, Loginov, Tsodikova,  Lezina, and Borzenko}

\OrigCopyrightedAuthors{P.Yu. Chebotarev,  A.K. Loginov, Ya.Yu. Tsodikova, Z.M. Lezina,  and V.I. Borzenko}

\received{Received April 6, 2009}

\OrigPages{pp.~18--25}

\maketitle

\begin{abstract}
Social dynamics determined by voting in a stochastic environment is analyzed for a society composed of two cohesive groups of similar size. Within the model of random walks determined by voting, explicit formulas are derived for the capital increments of the groups against the parameters of the environment and ``claim thresholds'' of the groups. The ``unanimous acceptance'' and ``unanimous rejection'' group rules are considered as the voting procedures. Claim thresholds are evaluated that are most beneficial to the participants of the groups and to the society as a whole.
\end{abstract}

\section{INTRODUCTION}

This work was carried out within the framework of analyzing social dynamics determined by democratic decisions in a stochastic environment. The subject of investigation is the relationship between the parameters of social dynamics and social strategies of the participants including egoism, collectivism (corporatism) and altruism. The results for the situation where a cohesive group competes with egoistic participants are reported in \cite{ch1,ch2}; competition between egoists and two cohesive groups (parties) was studied in~\cite{ch3}. The latter paper also considers the so-called ``snowball of cooperation'' mechanism which, under certain conditions, motivates those participants initially having egoistic strategies to act altruistically.

The present paper investigates a special case of the situation considered in \cite{ch3}, namely, competition of two groups. The relative simplicity of this case enables us to obtain explicit formulas for the most interesting relationships.

Let us recall the basic features of our model of voting in a stochastic environment. Being composed of a finite number of participants, the society sequentially votes for/against proposals that are regularly generated by the ``environment'' according to a random law. A proposal is identified with a vector of capital increments of the participants; an alternative interpretation involves the utility values of the participants. According to the model, the capital/utility increments entering the proposal are realizations of independent and identically distributed random variables. Here, the normal distribution is mainly considered. During the vote, an egoistic participant supports any proposal that increases his/her capital. In contrast to egoists, the members of a group vote jointly for those proposals that would be favorable for the whole group (the latter is characterized by a certain index). In particular, a group may support proposals that increase either the capital of the majority of its members or the total group capital. Each proposal is accepted (and implemented) or rejected under a particular voting procedure. Generally, the procedures of $\alpha $-majority are used. The parameter $\alpha\in[0,1)$ determines the share of votes necessary and sufficient for the acceptance of a proposal.

The presence of a stochastic environment is the focus of the model under consideration; it preserves many basic phenomena of social reality. At the same time, it enables the investigation of these phenomena by analytical methods.

In voting theory, the stochasticity assumption is often applied to choosing the voters' positions (see, e.g.,~\cite{ch4}). In cooperative game theory, it is widely adopted with respect to the payments~\cite{ch5}. A distinctive feature of our model is a stochastic mechanism for generating proposals for voting. In other words, we study a model of random walks controlled by voting. In this context, some papers on the dynamical correction of tax rates are relevant (see, e.g.,~\cite{ch6}); note, however, that these works focus on selecting the rate itself (which is optimized within special models of production and consumption) and not on the random walks in the space of individual utilities. We also mention a recent survey on modelling political competition processes that involve voting~\cite{ch7}.

\section{EXPECTED CAPITAL OF GROUP MEMBERS}

Suppose that a society consists of two groups  of similar size; every group votes jointly. A~natural voting rule requires that the accepted proposals are supported by both groups. In the sequel, we use the term \textit{unanimous acceptance group rule\/} for this voting procedure. Suppose that the first (second) group supports a proposal if and only if this proposal leads to the increase in the mean capital of the group members by at least $t_1$ (respectively, $t_2$). Here, the \emph{proposal support thresholds\/} $t_1$ and $t_2$ are variable parameters which may take positive, zero, or negative values. It is interesting to investigate the dependence of the future capital of the groups and the whole society on the proposal support thresholds $t_1$ and~$t_2$ (we will also call them \emph{claim thresholds}). In addition to the unanimous acceptance group rule, we will consider the rule according to which a proposal is accepted if and only if it is supported by at least one group. This voting procedure will be called the \textit{unanimous rejection group rule\/}, since a proposal is rejected iff it is rejected by both groups. In the following theorem, we derive expressions for the average capital increments under the above voting rules.

\begin{theorem}
Suppose a proposal is accepted if and only if it is supported by at least one group. Then the mathematical expectation of the one-step capital increment\/\footnote{Notation with tilde serves to indicate actual (i.e, taking into account the acceptance of proposals) capital increments, as distinct from the proposed ones.} of a member of group $i$ $(i=1,2)$ is given by the formula
\begin{gather}
\label{eq1}
M(\tilde d_i)
= \mu{\kern 1pt}F_{3-i}
+(\mu{\kern 1pt}F_i+\sigma_i f_i)\widebar F_{3-i}.
\end{gather}
Here$,$ $F_i=F(\mu_i/\sigma_i),\,\widebar F_i=1-F_i,\,f_i=f(\mu_i/\sigma_i),\,f(\cdot)$ and $F(\cdot)$ denote the density and distribution function of the standard normal distribution$,$ respectively$,$\footnote{%
$f(t)=\frac{1}{\sqrt{2\pi}}                    e^{-t^2\!/2}$ and
$F(t)=\frac{1}{\sqrt{2\pi}}\int_{-\infty}^{\,t}e^{-x^2\!/2}dx.$} $\mu_i=\mu-t_i,$ ${\sigma}_i=\sigma\!\!\Bigm/\!\!\!{\sqrt{g^{}_i}},$
$t_i$ is the claim threshold of group $i,$ $g^{}_i$ is the size of group $i,$ while $\mu$ and $\sigma$ are the parameters of the normal distribution $N(\mu,\,\sigma^2)$ that describes the independent capital increments forming proposals.

Suppose a proposal is accepted if and only if it is supported by both groups. Then the mathematical expectation of the one-step capital increment
of a member of group $i$ $(i=1,2)$ is given by the formula
\begin{gather}
\label{eq2}
M(\tilde d_i)
=(\mu{\kern 1pt}F_i+ \sigma_i f_i)F_{3-i} .
\end{gather}
\end{theorem}

\begin{proof}
For a given proposal, let $G_1$ ($G_2$) be the event that the first (respectively, second) group supports the proposal. Denote by $G_1 G_2$ the simultaneous implementation of $G_1$ and $G_2.$ Similarly, $\,G_1 \widebar G_2$ and $\widebar G_1\widebar G_2$ are the simultaneous implementation of $G_1$ and the complement of $G_2,$ and of the complements of $G_1$ and $G_2,$ respectively. Let $P(\cdot)$ denote the event probability.

To prove the first statement of the theorem, we first suppose that the unanimous rejection group rule is used.
Using the total probability formula one has
\begin{gather}
\label{eq3}
\begin{split}
 M(\tilde d_i)
=M(\tilde d_i \mid G_{3-i})\,P(G_{3-i})
+M(\tilde d_i \mid G_i\widebar G_{3-i})\,P(G_i \widebar G_{3-i})
+M(\tilde d_i \mid \widebar G_i \widebar G_{3-i})\,P(\widebar G_i \widebar G_{3-i}).
\end{split}
\end{gather}

On the assumption of $\widebar G_i \widebar G_{3-i},$ the proposal is rejected; hence, $M(\tilde d_i \!\mid\! \widebar G_i \widebar G_{3-i})=0$. Under $G_{3-i}$, the proposal is accepted and independence of the proposal components implies that ${M(\tilde d_i \!\mid\! G_{3-i})=\mu}$.
Similarly, the proposal is accepted under $G_i \widebar G_{3-i}$, and independence of the components $d_i$ yields $M(\tilde d_i \!\mid\! G_i \widebar G_{3-i})=M(d_i \!\mid\! G_i \widebar G_{3-i})=M(d_i \!\mid\! G_i)$. Next, $M(d_i \!\mid\! G_i)=M(d_i^{\,\textrm{ave}} \!\mid\! G_i)$ with $d_i^{\,\textrm{ave}}$ standing for the result of averaging the components of the proposal (these components correspond to the $i$th group). Observe that $d_i^{\,\textrm{ave}}$ is a random variable having the distribution $N(\mu,\sigma_i^2)$, where ${\sigma}_i=\sigma\!\!\Bigm/\!\!\!{\sqrt{g^{}_i}},$ and that $G_i$ is tantamount to $d_i^{\,\textrm{ave}}>t_i.$ Using the formula for the conditional mean of a normal random variable (it can be proved by integration) one finally obtains
\begin{gather}
\label{eq4}
M(\tilde d_i \mid G_i \widebar G_{3-i})
=M(d_i^{\,\textrm{ave}} \mid d_i^{\,\textrm{ave}} >t_i)
={\mu}+\frac{\sigma_i f_i}{F_i},
\end{gather}
where $F_i=F(\mu_i/\sigma_i)$, $f_i=f(\mu_i/\sigma_i)$, and $\mu_i=\mu-t_i.$ It is also easy to demonstrate that ${P(G_{3-i})=F_{3-i}}$. Using independence of the proposal components we have the expression
\begin{gather}
\label{eq5}
\,P(G_i \widebar G_{3-i})=\,P(G_i)\,P(\widebar G_{3-i})=F_i \widebar F_{3-i} .
\end{gather}
Substituting all this into\,(\ref{eq3}) yields
\begin{gather*}
M(\tilde d_i)
=\mu{\kern 1pt}F_{3-i}+\left(\mu+\frac{\sigma_i f_i}{F_i}\right)F_{i{\kern 1pt}}\widebar F_{3-i}
=\mu{\kern 1pt}F_{3-i}+(\mu{\kern 1pt}F_{i{\kern 1pt}}+\sigma_i f_i){\kern 1pt} \widebar F_{3-i}.
\end{gather*}
This completes the proof of the first part.

Now prove the second statement of the theorem. Suppose that the unanimous acceptance group rule is used. In this case,
\begin{gather}
\label{eq6}
M(\tilde d_i)=M(\tilde d_i \mid G_i G_{3-i})\,P(G_i G_{3-i}).
\end{gather}

Similarly to the derivation of (\ref{eq4}), one obtains
\begin{gather}
\label{eq7}
M(\tilde d_i\mid G_i G_{3-i})
=M(d_i^{\,\textrm{ave}}\mid d_i^{\,\textrm{ave}}>t_i)=\mu+\frac{\sigma_i f_i}{F_i}
\end{gather}
and similarly to (\ref{eq5}),
\begin{gather}
\label{eq8}
\,P(G_i G_{3-i})=\,P(G_i)\,P(G_{3-i})=F_i F_{3-i} .
\end{gather}

Substitute (\ref{eq7})--(\ref{eq8}) in (\ref{eq6}) to get $M(\tilde d_i)=(\mu{\kern 1pt}F_i+\sigma_i f_i)\,F_{3-i}.$ This completes the proof.~$\diamondsuit$
\end{proof}

Of special interest are the \textit{comparative\/} capital increments of the groups. Indeed, it is important to know the capital-related performance of each group against the background of the results achieved by the other one. The expected one-step capital increment of Group~1 as compared with the one of Group~2 is expressed as $M(\tilde d_1-\tilde d_2)$. Theorem~1 enables one to obtain simple formulas for $M(\tilde d_1-\tilde d_2)$.

\begin{corollary}
In the notation of Theorem~$1,$ if the unanimous rejection group rule is adopted$,$ then
\begin{gather}
\label{eq9}
M(\tilde d_1-\tilde d_2)
=\sigma_1 {\kern 1pt}f_1 {\kern 1pt}\widebar F_2-\sigma_2{\kern 1pt}f_2{\kern 1pt}\widebar F_1 .
\end{gather}

Under the unanimous acceptance group rule$,$ one has
\begin{gather}
\label{eq10}
M(\tilde d_1-\tilde d_2)
=\sigma_1 {\kern 1pt}f_1{\kern 1pt}F_2
-\sigma_2 {\kern 1pt}f_2{\kern 1pt}F_1.
\end{gather}
\end{corollary}

The proof of Corollary~1 is straightforward.

\section{ANALYSIS OF SOCIAL DYNAMICS}

In order to apply and interpret the results obtained, let us first consider the case where the mean proposed capital increment is zero, i.e., the environment is neutral; suppose that the standard deviation of the proposed capital increment is~10: $\mu=0,\, \sigma=10$. Suppose that each group includes 300 participants and Group~1 supports a proposal whenever it increases the total capital of the group $(t_1=0)$. The claim threshold $t_2$ of Group~2 varies.

In the case of the unanimous acceptance group rule (sometimes, we will call it the $(G_1\wedge G_2)$-rule), the expected capital increments of the groups and the whole society versus the claim threshold $t_2$ of Group~2, as stated by Theorem~1, are depicted in Fig.\,1.

\begin{figure}[htbp]
\centerline{\includegraphics[scale=1.5]{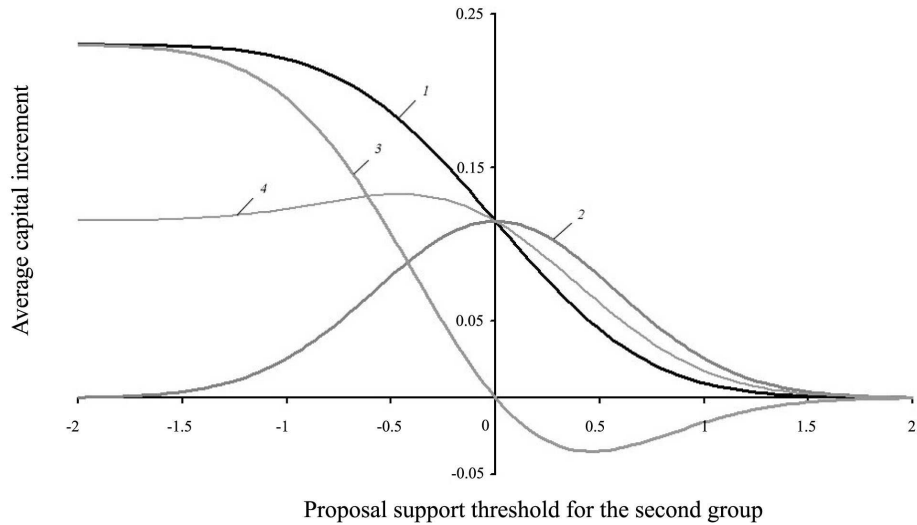}}
\label{fig1}
\caption{Expected capital increments of the groups and the society provided that the unanimous acceptance group rule is adopted (in this example, it coincides with the majority rule): {\sl1}---Group~1; {\sl2}---Group~2; {\sl3}---the difference between Group~1 and Group~2; {\sl4}---the whole society. The parameters are: $\mu=0;\, \sigma=10;$ ${g_1=g_2=300}$.}
\end{figure}

%
%

In particular, the average capital increment of Group~2 represents a symmetrical bell-shaped curve. This can be explained as follows. Setting a positive Group~2 threshold (high claims) leads to the rejection of some favorable to Group~2 (on the average) proposals, whereas setting a negative threshold (low claims) leads to the acceptance, along with the favorable to Group~2 proposals, of some proposals unfavorable to it. So both strategies decrease the expected capital increment of Group~2 as against the case of~$t_2=0.$ A~noteworthy fact is that positive and negative thresholds of the same absolute value are equivalent. Let us prove it.

\begin{corollary}
Let $\mu=0$. Then for fixed values of the model parameters except for~$t_2,$ any $t\in\bbbr,$ and under either group voting rule$,$ $M(\tilde d_2 \mid t_2=t)=M(\tilde d_2 \mid t_2=-t)$ holds.
\end{corollary}

\begin{proof}
Applying Theorem~1 to the case $\mu=0$ one obtains
\begin{gather}
\label{eq11}
M(\tilde d_2)=\sigma_2 f_2 \,\widebar F_1
\end{gather}
for the unanimous rejection group rule and
\begin{gather}
\label{eq12}
M(\tilde d_2)=\sigma_2 f_2 \,F_1
\end{gather}
for the unanimous acceptance group rule.

The required statement follows from the facts that $\sigma_2$ and $F_1$ are independent of $t_2$ and that $f_2=f\left({\dfrac{\mu-t_2}{\sigma_2}}\right)$ is an even function of $t_2$ provided that $\mu=0$.~$\diamondsuit$
\end{proof}

\begin{remark}
Due to (\ref{eq11})--(\ref{eq12}), for both voting rules under consideration, the expression of $M(\tilde d_2)$ in terms of $t_2$ is proportional to the normal density function, which is a bell-shaped curve with zero limits as $t_2 \to-\infty$ and $t_2 \to \infty.$ Finally, if $t_1=0$, then under \textit{both voting rules\/}, $M(\tilde d_2)=\frac{1}{2}\sigma_2 f_2$ holds.~$\diamondsuit$
\end{remark}

Next we study the characteristics of social dynamics for each voting rule separately.

\subsection{The Unanimous Acceptance Group Rule}

If the unanimous acceptance group rule is adopted, then {Group~1} benefits from low (negative) claim threshold $t_2$ of Group~2. Indeed, in this case Group~2 rarely vetoes proposals beneficial to Group~1 and so Group~1 can maximize its benefits. Contrariwise, if $t_2$ is high, then Group~1 can rarely expect approval of proposals beneficial to it.

Now, suppose that Group~2 aims not at maximizing the average capital increment, but at ensuring maximum advantage over Group~1. In other words, Group~2 maximizes $M(\tilde d_2-\tilde d_1)$ rather than $M(\tilde d_2)$. A~tool Group~2 can use for this is its claim threshold~$t_2$. For instance, in the example illustrated by Fig.\,1, $M(\tilde d_2-\tilde d_1)$ as a function of $t_2$ attains its maximum when $M(\tilde d_1-\tilde d_2)$  attains its minimum, i.e., at $t_2\approx 0.46.$ In the proposition below, we find the ``claim threshold'' $t_2$ that maximizes $M(\tilde d_2-\tilde d_1)$ in the general case.

\begin{proposition}
Suppose the unanimous acceptance group rule is adopted. Then$,$ in the notation of Theorem~$1,$  the expected advantage $M(\tilde d_2-\tilde d_1)$ of a~Group~$2$ member over a~Group~$1$ member in the sense of capital increment attains its unique maximum at the claim threshold $t_2$ of Group~$2$ defined by
\begin{gather}
\label{eq13}
t_2^+=\mu+\frac{\sigma_1 f_1}{F_1}.
\end{gather}
\end{proposition}

\begin{proof}
Differentiating (10) with respect to $t_2$ we obtain
\begin{gather*}
\frac{d{\kern 1pt}M(\tilde d_1-\tilde d_2)}{d\,t_2}
=\frac{d\,(\sigma_1 {\kern 1pt}f_1 {\kern 1pt}F_2-\sigma_2 {\kern 1pt}f_2 {\kern 1pt}F_1)}{d\,t_2}
=\sigma_1 {\kern 1pt}f_1 f_2\!\cdot\!\bigl(-\sigma_2^{-1}\bigr)-\sigma_2 f_2 \frac{-(t_2-\mu)}{\sigma_2^2}F_1 .
\end{gather*}

Setting the derivative equal to zero results in the desired expression~(\ref{eq13}). Finding the second derivative of the function $M(\tilde{d}_1-\tilde d_2)$ ensures that a minimum has been found. Hence, (\ref{eq13}) provides the unique maximum of $M(\tilde d_2-\tilde d_1)$.~$\diamondsuit$
\end{proof}

\begin{remark}
It should be noted that the optimal claim threshold $t_2^+$ found in Proposition~1 is independent of the group size $g_2.$ It is also interesting that expression (\ref{eq13}) coincides with (\ref{eq4}) in which $i=1$. In turn, (\ref{eq4}) with $i=1$ expresses the expected capital increment of a member of Group~1 provided it exceeds the threshold $t_1$. This observation suggests the following simple algorithm of estimating the optimal ``claim threshold'' $t_2^+$ by Group~2: (1) find the average capital increments of Group~1 members according to the proposals that Group~1 supports; (2)~set the threshold $t_2$ equal to the value found on Step~(1) averaged over the proposals.

Thus, to achieve an advantage over Group~1, Group~2 should set its claim threshold at a higher level than Group~1 does. As has been noted before, the optimal threshold $t_2$ is equal to the mean ``above threshold'' value of Group~1 capital increment. In particular, for the example of Fig.\,1, (\ref{eq13}) gives $t_2^+=\sqrt{\frac{2}{3\pi}}.$ Now, if Group~1 wishes to act in the same way (i.e., to maximize its advantage $M(\tilde d_1-\tilde d_2)$ with respect to $t_2=t_2^+$), it should, in turn, set a higher threshold: $t_1>t_2^+.$ If this process continues, it leads to an infinite increase in the claims (the process is divergent and the corresponding game possesses no Nash equilibrium).~$\diamondsuit$
\end{remark}

How does the maximum comparative gain $M(\tilde d_2-\tilde d_1)$ of Group~2 depend on the fixed claim threshold $t_1$ of Group~1? It increases as $t_1 \to-\infty$, since in this case Group~1 supports all proposals, and thus, the decisions are made by Group~2. If $t_1$ grows, then this maximum gain decreases and tends to zero as $t_1\to\infty$ since in this case, no proposals are accepted due to the veto of Group~1.

Now consider the following question which seems quite interesting. What claim threshold $t_2$ of Group~2 is optimal for the whole society (thus, leading to the maximum expected capital increment of the society)? The answer is provided by the following proposition.

\begin{proposition}
Suppose the unanimous acceptance group rule is adopted. Then the expected capital increment of the whole society attains its maximum value at the claim threshold $t_2$ of Group~$2$ given by
\begin{gather}
\label{eq14}
t_2^0
=-\frac{g_1}{g_2}\left({\mu+\frac{\sigma_1 f_1}{F_1}}\right).
\end{gather}
\end{proposition}

\begin{proof}
Theorem~1 (see (\ref{eq2})) implies that the expected capital increment of the whole society is
\begin{eqnarray*}
g_1 M(\tilde d_1)+g_2 M(\tilde d_2)
&=&g_1 (\mu{\kern 1pt}F_1+\sigma_1 f_1)F_2+g_2(\mu{\kern 1pt}F_2+\sigma_2 f_2)F_1 \\
&=&\left({(g_1+g_2)\mu{\kern 1pt}F_1+g_1 \sigma_1 f_1}\right)F_2+g_2 \sigma_2 f_2 F_1.
\end{eqnarray*}

The derivative of this function with respect to $t_2$ is
\begin{gather*}
\frac{d\,(g_1 M(\tilde d_1)+g_2 M(\tilde d_2))}{d\,t_2}
=((g_1+g_2)\mu{\kern1pt}F_1+g_1\sigma_1f_1)f_2\!\cdot\!\bigl(-\sigma_2^{-1}\bigr)+g_2\sigma_2{\kern1pt}f_2F_1\frac{-(t_2-\mu)}{\sigma_2^2}.
\end{gather*}

Setting it equal to zero one obtains $g_2 {\kern 1pt}F_1 {\kern 1pt}(t_2^0-\mu)=-(g_1+g_2)\mu{\kern 1pt}F_1-g_1\sigma_1 f_1$, which leads to~(\ref{eq14}). Finally, the second derivative test confirms that this point provides a unique maximum value.~$\diamondsuit$
\end{proof}

\begin{remark}
\label{r_3}
Comparing Propositions~1 and~2 produces a somewhat unexpected result. Namely, at $g_1=g_2$ those values of claim threshold $t_2$ of Group~2 leading to the maximum advantage over Group~1 (on the one hand) and to the best results for the whole society (on the other hand) turn out to be \textit{opposite.} In particular, this is the case for the example illustrated by Fig.\,1. Note that even if $g_1\ne g_2$, the thresholds $t_2^+$ and $t_2^0$ always have different signs (or both are zero), since by (\ref{eq13}) and (\ref{eq14}), 
\begin{gather*}
t_2^0=-\frac{g_1}{g_2}\,t_2^+.\quad\diamondsuit
\end{gather*}
\end{remark}

Proposition~2 determines $t_2$ (as a function of $t_1$) that guarantees the maximum capital increment of the whole society. Let us maximize this maximum over $t_1$ to find the global optimal point of the society. Obviously, the corresponding solutions $t_1$ and $t_2$ satisfy the system of equations
\begin{gather}
\label{eq15}
t_1=-\frac{g_2}{g_1}\left({\mu+\frac{\sigma_2 f_2}{F_2}}\right); \quad
t_2=-\frac{g_1}{g_2}\left({\mu+\frac{\sigma_1 f_1}{F_1}}\right)
\end{gather}
with $\sigma_1 ,\sigma_2 ,f_1 ,f_2 ,F_1,$ and $F_2$ defined in Theorem~1. Most likely, the solution of (\ref{eq15}) cannot generally be expressed in elementary functions. In particular, this applies even to the simplest case of $g_1=g_2$ (two groups of equal size) in which the system (\ref{eq15}) reduces to
\begin{gather*}
t_1=-\mu-\frac{\sigma_1 f_1}{F_1}; \quad t_2=t_1.
\end{gather*}

Nevertheless, the solution can be obtained numerically. For instance, the dependence on $\mu$ of the claim threshold $t=t_1=t_2$ optimal for the society is shown (in the case of $g_1=g_2=300$ and $\sigma=10$) in Fig.\,2.

\begin{figure}[htbp]
\centerline{\includegraphics[scale=1.5]{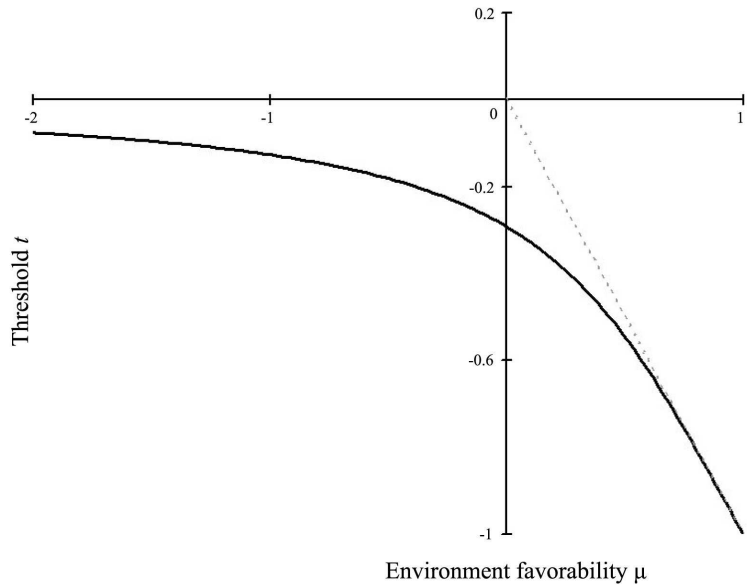}}
\label{fig2}
\caption{Relationship between $\mu$ and the claim threshold $t=t_{1}=t_{2}$ that maximizes the capital increment of the whole society provided that $g_{1}=g_{2}=300$ and $\sigma=10$.}
\end{figure}

%
%

In particular, for $g_1=g_2=\tilde g$ and $\mu=0,$ the optimal threshold is $t=-\sigma y_0\big/\sqrt{\tilde g}$, where $y_0\approx 0.506$ is the unique solution of the equation $y=f(y)/F(y)$ with $f(y)$ and $F(y)$ designating the density and distribution function of the standard normal distribution.

One of the most instructive results of this study is that under the unanimous acceptance group rule, the claim threshold $t$ of the groups that is optimal for the whole society is negative (see Fig.\,2). In the case $g_1=g_2$, it tends to zero from below as $\mu\to-\infty$ and asymptotically approaches to the line $t=-\mu$ as $\mu\to\infty$. Negative claim thresholds of the groups can be interpreted as a willingness to a moderately negative result, that is, as a certain appetite for risk.

\subsection{The Unanimous Rejection Group Rule}

Now consider the unanimous rejection group rule (we will also call it the $(G_1\vee G_2)$-rule). With this procedure, for acceptance of a proposal, the support of either one group is sufficient. The analytical results obtained in Section~2 are illustrated by the diagram in Fig.\,3.

\begin{figure}[htbp]
\centerline{\includegraphics[scale=1.5]{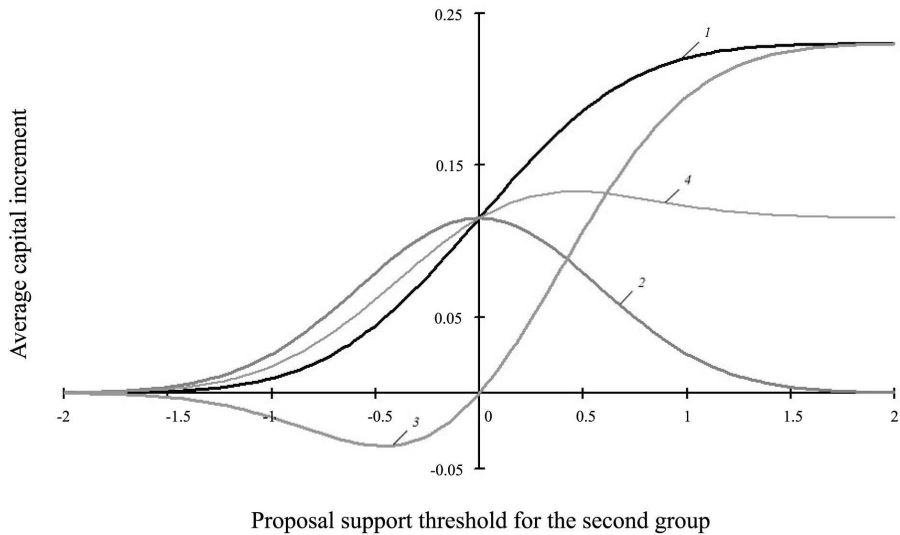}}
\label{fig3}
\caption{Expected capital increments of the members of Groups~1 and~2 and the society under the unanimous rejection group rule: {\sl1}---Group~1; {\sl2}---Group~2; {\sl3}---the difference between Group~1 and Group~2; {\sl4}---the whole society; the parameters are: $\mu=0;\, \sigma=10;$ $g_{1}=g_{2}=300$.}
\end{figure}

%
%

It can be easily observed that this diagram mirrors the one shown in Fig.\,1. Indeed, the following statement holds.

\begin{corollary}
Let $\mu=0$. If the model parameters$,$ except for $t_2,$ are fixed$,$ then for any $t\in\bbbr,$ $M_{G_1\vee G_2}(\tilde d_1\mid t_2=t)=M_{G_1\wedge G_2}(\tilde d_1\mid t_2=-t)$ holds$,$ where $M_{G_1 \wedge G_2}$ $(M_{G_1\vee G_2})$ is the expectation$,$ provided the unanimous acceptance group rule $($resp.$,$ the unanimous rejection group rule$)$ is adopted. Moreover$,$ at $t_1=0$ one has $M_{G_1 \vee G_2}(\tilde d_2\mid t_2=t)=M_{G_1 \wedge G_2}(\tilde d_2\mid t_2=t)$ for any $t\in\bbbr.$
\end{corollary}

\begin{proof}
According to Theorem~1, for $\mu=0$ we have:
\begin{eqnarray*}
M_{G_1 \wedge G_2} (\tilde d_1\mid t_2=-t)
                                            &=&\sigma_1 f_1{\kern 1pt}F(t/\sigma_1),\\
M_{G_1 \vee   G_2} (\tilde d_1\mid t_2=t)
                                            &=&\sigma_1 f_1{\kern 1pt}(1-F(-t/\sigma_1))
                                             = \sigma_1 f_1{\kern 1pt}   F( t/\sigma_1)
=M_{G_1\wedge G_2}(\tilde d_1\mid t_2=-t).
\end{eqnarray*}

Moreover, $t_1=0$ implies that $M_{G_1\wedge G_2}(\tilde d_2\mid t_2=t)=\sigma_2{\kern 1pt}f(-t/\sigma_2)F(0)=M_{G_1 \vee G_2}(\tilde d_2\mid t_2=t)$. This completes the proof.~$\diamondsuit$
\end{proof}

\begin{remark}
Due to Corollary~2, $\mu=0$ implies that $M(\tilde d_2\mid t_2=t)=M(\tilde d_2\mid t_2=-t)$ for either voting rule under study. Consequently, if $\mu=0$ and $t_1=0,$ then every relationship under $(G_1\vee G_2)$-rule can be derived from the corresponding relationship obtained for the $(G_1\wedge G_2)$-rule by means of the $y$-axis reflection. Thus, in this case, the situations where voting is organized according to the $(G_1\wedge G_2)$-rule and according to the $(G_1\vee G_2)$-rule are, in a certain sense, dual.~$\diamondsuit$
\end{remark}

Note that under the $(G_1\vee G_2)$-rule, Group~1 benefits from excessive claims of Group~2 (i.e., from high threshold~$t_2$). In this case, if threshold $t_1$ is positive, but not so high, then most decisions are made for the benefit of Group~1. On the other hand, at a low threshold $t_2,$ Group~2 ensures implementation of almost all proposals, and so most proposals unfavorable to Group~1 pass through the ``riddle'' of voting.

Now let us consider the maximization problem for the Group~2 advantage $M(\tilde d_2-\tilde d_1)$ over Group~1 provided that the $(G_1\vee G_2)$-rule is adopted.

\begin{proposition}
Under the $(G_1\vee G_2)$-rule$,$ the expected advantage $M(\tilde d_2-\tilde d_1)$ of a Group~$2$ member over a Group~$1$ member attains its maximum value at the claim threshold $t_2$ defined by
\begin{gather}
\label{eq16}
t_2^+=\mu-\frac{\sigma_1 f_1}{\widebar F_1}.
\end{gather}
\end{proposition}

\begin{proof}
Differentiating (\ref{eq9}) with respect to $t_2$ yields:
\begin{gather*}
\frac{d{\kern 1pt}M(\tilde d_1-\tilde d_2)}{d\,t_2}
=\frac{d\,(\sigma_1{\kern 1pt}f_1{\kern 1pt}\widebar F_2-\sigma_2{\kern 1pt}f_2{\kern 1pt}\widebar F_1)}{d\,t_2}
=\sigma_1 {\kern 1pt}f_1 f_2\!\cdot\!\bigl(\sigma_2^{-1}\bigr)-\sigma_2 f_2 \frac{-(t_2-\mu)}{\sigma_2^2}\widebar F_1.
\end{gather*}

Setting this derivative equal to zero one arrives at~(\ref{eq16}).
The second derivative test confirms that this unique extremum is a minimum.
Thus, (\ref{eq16}) provides a maximum value to the function $M(\tilde d_2-\tilde d_1)$.~$\diamondsuit$
\end{proof}

The maximum comparative gain $M(\tilde d_2-\tilde d_1)$ of Group~2 attained at the claim threshold $t_2^+$ tends to zero as $t_1\to-\infty$ (in the limit, all proposals are accepted by the votes of Group~1) and tends to its maximum value when $t_1 \to \infty$ (in the limit, Group~1 does not support any proposals).

As well as in the case of the $(G_1\wedge G_2)$-rule, consider the problem of maximizing, by choosing $t_2$, the capital of the whole society. The solution is provided by the following proposition.

\begin{proposition}
Suppose that the $(G_1\vee G_2)$-rule is adopted. Then the expected capital increment of the whole society attains its maximum value at the Group~$2$ claim threshold given by
\begin{gather}
\label{eq17}
t_2^0=-\frac{g_1}{g_2}\left({\mu-\frac{\sigma_1 f_1}{\widebar F_1}}\right).
\end{gather}
\end{proposition}

\begin{proof}
By Theorem~1 (see Eq.\,(\ref{eq1})) the expected capital increment of the whole society is
\begin{eqnarray*}
g_1 M(\tilde d_1)+g_2 M(\tilde d_2)
&=&g_1\!\cdot\!\left({\mu{\kern 1pt}F_2+(\mu{\kern 1pt}F_1+{\sigma}_1 f_1)\widebar F_2}\right)
 + g_2\!\cdot\!\left({\mu{\kern 1pt}F_1+(\mu{\kern 1pt}F_2+{\sigma}_2 f_2)\widebar F_1}\right)\\
&=&(g_1+g_2)\mu{\kern 1pt}F_1+g_1\sigma_1 f_1+g_2\sigma_2\widebar F_1 f_2
+\left(g_1\!\cdot\!(\mu{\kern 1pt}-\mu{\kern 1pt}F_1-\sigma_1 f_1)+g_2\mu{\kern 1pt}\widebar F_1\right)\!F_2.
\end{eqnarray*}

The derivative of this function w.\,r.\,t. $t_2$ is:
\begin{gather*}
\frac{d\,(g_1 M(\tilde d_1)+g_2 M(\tilde d_2))}{d\,t_2}
=g_2 \sigma_2 \widebar F_1 f_2\frac{-(t_2-\mu{\kern1pt})}{\sigma_2^2}
+\left({g_1\!\cdot\!(\mu{\kern 1pt}-\mu{\kern 1pt}F_1-\sigma_1 f_1)+g_2\mu{\kern 1pt}\widebar F_1}\right)\!f_2\!\cdot\!\bigl(-\sigma_2^{-1}\bigr).
\end{gather*}

Setting it equal to $0$ we have: $-g_2{\kern 1pt}\widebar F_1{\kern 1pt}(t_2^0-\mu)=g_1\!\cdot\!(\mu{\kern 1pt}\widebar F_1-\sigma_1 f_1)+g_2\mu{\kern 1pt}\widebar F_1,$ which leads to~(\ref{eq17}). Finally, using the second derivative test we verify that the point in question provides the maximum value.~$\diamondsuit$
\end{proof}

\begin{remark}
Comparing Propositions~3 and~4 we obtain the following (cf.\ Remark~\ref{r_3}). At $g_1=g_2$, the value $t_2^+$ of the claim threshold $t_2$ of Group~2 that leads to the maximum advantage over Group~1 and the value $t_2^0$ that maximizes the capital of the whole society are \textit{opposite.} In particular, this is the case for the example in Fig.\,3. Moreover, even if $g_1\ne g_2$, the thresholds $t_2^+$ and $t_2^0$ have different signs (or both are equal to zero), since, as well as in the case of the  $(G_1\wedge G_2)$-rule, ${t_2^0=-(g_1 /g_2)t_2^+}$.~$\diamondsuit$
\end{remark}

Now consider the issue of maximizing the capital increment of the whole society when the $(G_1\vee G_2)$-rule is employed. In a certain sense, the result is dual to the one derived for the $(G_1 \wedge G_2)$-rule. Namely, Proposition~4 implies that the claim thresholds of the groups, $t_1$ and $t_2$, that maximize the capital increment of the whole society satisfy the following system of equations:
\begin{gather*}
t_1=-\frac{g_2}{g_1}\left({\mu-\frac{\sigma_2 f_2}{\widebar F_2}}\right); \quad
t_2=-\frac{g_1}{g_2}\left({\mu-\frac{\sigma_1 f_1}{\widebar F_1}}\right),
\end{gather*}
where $\sigma_1,\sigma_2,f_1,f_2,F_1,$ and $F_2$ are defined in Theorem~1. In the simplest case $g_1=g_2,$ the above system of equations reduces to
\begin{gather*}
t_1=-\mu+\frac{\sigma_1 f_1}{\widebar F_1};\quad t_2=t_1.
\end{gather*}

The dependence on $\mu$ of the claim threshold $t$ of both groups that maximizes the capital of the society is illustrated (for the case of $g_1=g_2=300$ and $\sigma=10$) in Fig.\,4. In particular, for $g_1=g_2=\tilde g$ and $\mu=0,$ the corresponding threshold is $t=\sigma y_0\big/\sqrt{\tilde g},$ where $y_0\approx 0.506$ is the solution of the equation $y=f(y)/F(y)$, $f(y)$ and $F(y)$ being the density and distribution function of the standard normal distribution, respectively.
\begin{figure}[htbp]
\centerline{\includegraphics[scale=1.5]{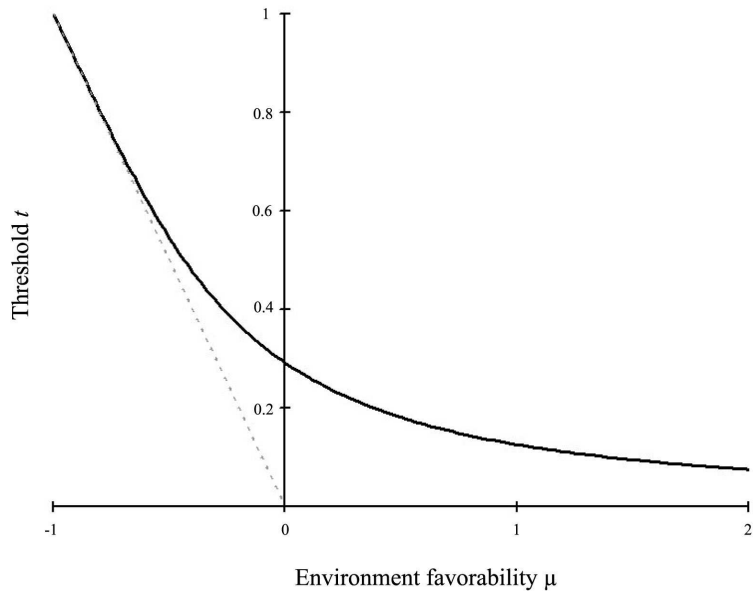}}
\label{fig4}
\caption{Relationship between the parameter $\mu$ and the claim threshold $t=t_{1}=t_{2}$ of both groups that maximizes the capital increment of the whole society in the case of $g_{1}=g_{2}=300$ and $\sigma=10$.}
\end{figure}

In contrast to voting based on the $(G_1\wedge G_2)$-rule, the claim threshold of both groups that maximizes the capital of the society under the $(G_1\vee G_2)$-rule is positive (if the voting rule is not conservative, then the voters have to be conservative); for $g_1=g_2$ it tends to zero from above as $\mu\to\infty$ and asymptotically approaches to the line $t=-\mu$ as $\mu\to-\infty$ (see Fig.\,4).

\section{Conclusion}

This paper has studied the social dynamics determined by voting of two cohesive groups in a stochastic environment. Two groups of similar size (a kind of a ``two-party system'') have been considered; this makes reasonable the application of the voting procedure which requires the support of both groups for the acceptance of any proposal (the unanimous acceptance group rule).

Another procedure in which the support of either group suffices for the acceptance of a proposal (the unanimous rejection group rule) looks a little less natural. However, using this procedure does not lead to an infinite sequence of contradictory decisions suggested by different groups since the proposals in the model are generated by the stochastic environment. 

Each group supports only those proposals that increase the average capital of its members by an amount exceeding the claim threshold set by the group. The relationship between social dynamics and the claim thresholds of the groups was the main subject of the present study.

In particular, we have derived analytical expressions for the capital increments of the groups in terms of their claim thresholds. It has been shown that under the unanimous acceptance group rule and the equal size of the groups, the claim threshold of a group that ensures the maximum advantage over the other group and the claim threshold that maximizes the capital of the whole society are \emph{opposite}. Moreover, the first above threshold is equal to the average capital increment of the competing group members over the proposals supported by the competing group. Furthermore, it turns out that the claim thresholds of the groups that maximize the capital of the whole society are negative, which can be interpreted as the profitability of a certain appetite for risk in the case of a conservative voting rule; the dependence of these thresholds on the environment favorability has been established. Similar (in fact, dual) results have been obtained for the unanimous rejection group rule.

It should be emphasized that the phenomena observed are not only inherent in the model under consideration; the most important ones have real prototypes widely discussed in the relevant literature.

\end{document}